\pdfoutput=1
\documentclass[aps,pra,preprint,floatfix,amsmath,amssymb,longbibliography]{revtex4-2}
\usepackage{amsmath,amssymb,amsthm,graphicx,hyperref}
\DeclareTextSymbolDefault{\textbullet}{OMS}
\DeclareTextSymbol{\textbullet}{OMS}{15}
\graphicspath{{figures/}{./}}
\hypersetup{colorlinks=true,linkcolor=blue,citecolor=blue,urlcolor=blue}

\newtheorem{theorem}{Theorem}[section]
\newtheorem{proposition}[theorem]{Proposition}
\newtheorem{lemma}[theorem]{Lemma}
\newtheorem{corollary}[theorem]{Corollary}
\theoremstyle{definition}
\newtheorem{remark}[theorem]{Remark}
\newcommand{\TV}{\operatorname{TV}}
\newcommand{\Tr}{\operatorname{Tr}}
\newcommand{\qmi}{I(A:B)}
\newcommand{\cmiq}{I(Q_A:Q_B)}
\newcommand{\cmir}{I(R_A:R_B)}
\newcommand{\lmax}{\lambda_{\max}}

\begin{document}

\title{The complementary quantum correlations conjecture: Counterexamples and sufficient conditions}
\author{Teng Zhang}
\email{zhangteng@gxu.edu.cn}
\affiliation{Guangxi University, Nanning, China}
\author{Qian Lilong}
\email{qian.lilong@u.nus.edu}
\affiliation{Independent Researcher}
\date{\today}

\begin{abstract}
The complementary quantum correlations (CQC) conjecture says that the sum
of the mutual informations from two complementary local measurement settings
cannot exceed the state's quantum mutual information. We find counterexamples
in every dimension pair with both local dimensions at least two, except
\(2\times2\), where CQC is known to hold, and \(2\times3\), which remains open
up to exchanging the parties. At \(2\times5\), a rigorous interval check
certifies a full-rank example with excess greater than \(0.0002\) nats.
In all other covered dimensions, we construct rank-two separable examples;
a little white noise gives full-rank separable examples there as well.
On the positive side, CQC holds when Alice is a qubit and either complementary
readout is unbiased, even for arbitrary POVMs on Bob. We also bound the gap
for a specific \(2\times3\) family and give a corrected inequality valid in
all dimensions. Simulations illustrate why random searches can miss the
qutrit counterexamples.
\end{abstract}

\maketitle

\section{Introduction}

Can two complementary measurements reveal more correlation than the quantum
mutual information of the state being measured? The complementary quantum
correlations conjecture (CQC) of Schneeloch, Broadbent, and
Howell~\cite{Schneeloch2014} says no. For a bipartite state \(\rho_{AB}\),
choose mutually unbiased bases \(Q_A,R_A\) on Alice and \(Q_B,R_B\) on Bob.
The conjectured bound is
\begin{equation}
 \cmiq+\cmir\leq\qmi_\rho. \tag{CQC}\label{eq:cqc}
\end{equation}
The left side adds the classical mutual information from two separate
measurement settings: both parties choose \(Q\), or both choose \(R\).
The right side is the quantum mutual information before either measurement.

\paragraph{When does CQC fail?}
The answer depends on the local dimensions. Up to exchanging the parties,
the classification for dimensions at least two is:
\begin{center}
\begin{tabular}{ll}
\hline
Dimensions & Status of CQC \\
\hline
\(2\times2\) & Holds for all states~\cite{Wang2026} \\
\(2\times3\) & Open in general \\
\(2\times5\) & Full-rank counterexample (interval check) \\
All other pairs & Rank-two separable counterexamples \\
\hline
\end{tabular}
\end{center}
If either dimension is one, all three mutual informations are zero.
Wang, Wang, and Chen independently found counterexamples in equal dimensions
\(d\ge3\) and posted first~\cite{Wang2026cx}. Our constructions also cover
unequal dimensions. The \(2\times5\) result uses a fixed rational matrix and
a rigorous interval check, not a floating-point sign test. Appendix~\ref{app:proofs}
gives the other constructions, with interval checks at the two tail endpoints.

\paragraph{Why can it fail?}
A simple example uses a fair random bit to choose between two product states.
In the computational basis each party can recover the bit perfectly. This
already accounts for all of the state's quantum mutual information. The
Fourier measurements still reveal some of the same bit, so adding the two
measured mutual informations counts part of that correlation twice
(Sec.~\ref{sec:universal}).

\paragraph{When does it still hold?}
If Alice is a qubit and either complementary measurement gives two equally
likely outcomes, CQC holds even when Bob uses arbitrary POVMs
(Sec.~\ref{sec:unbiased-qubit}). We also prove CQC for a specific family of
equal mixtures of two product states at \(2\times3\), using an analytic
reduction and an interval check (Sec.~\ref{sec:closure}). This does not settle
the general \(2\times3\) case. Without those restrictions, a state- and
measurement-dependent correction gives a universal bound
(Sec.~\ref{sec:corrected}).

Earlier random searches found no violations~\cite{Schneeloch2014,Iqbal2026}.
Section~\ref{sec:sampling} illustrates why the tested random ensembles can
miss qutrit counterexamples. The simulations and extrapolations there are
not bounds on the probability of finding a violation.

\section{Definitions}\label{sec:defs}

For a joint table \(p=(p_{ij})\) write
\begin{equation}
 I(p)=\sum_{i,j:\,p_{ij}>0}p_{ij}\log\frac{p_{ij}}{p_{i\bullet}p_{\bullet j}},
 \qquad 0\log0:=0.
\end{equation}
Two orthonormal bases \(\{|q_j\rangle\}\) and \(\{|r_k\rangle\}\) of
\(\mathbb C^d\) are mutually unbiased when \(|\langle q_j|r_k\rangle|^2=1/d\)
for all \(j,k\). Basis vectors are columns, measurement operators are
\(|b_j\rangle\!\langle b_j|\), and the tensor ordering is \(A\otimes B\), so
\begin{equation}
 p^Q_{ij}=\Tr\!\left[(|q_i\rangle\!\langle q_i|\otimes
 |q_j\rangle\!\langle q_j|)\,\rho_{AB}\right].
\end{equation}
Throughout, \(H(\cdot)\) is the Shannon entropy of a probability vector,
\(H_2(u)=H(u,1-u)\), and \(H(Q_A)\) denotes the entropy of the outcome
distribution of measuring \(Q_A\); \(S(\cdot)\) is von Neumann entropy. For a
symmetric matrix, \(e_k\) denotes the \(k\)-th elementary symmetric function
of its eigenvalues. The \emph{CQC gap} of a state and a measurement choice is
\begin{equation}
 G(\rho)=\qmi_\rho-\cmiq-\cmir,
\end{equation}
so \(G<0\) is a strict counterexample. We write
\(F_d\) for the canonical Fourier basis,
\((F_d)_{jk}=\omega^{jk}/\sqrt d\) with \(\omega=e^{2\pi i/d}\), and
\(h(t)=H\!\left(\frac{1+t}2,\frac{1-t}2\right)\).

We use the following binary-entropy identities and bounds. For \(0\le u\le1\),
\begin{equation}
 g(u)=\log 2-h(\sqrt u)=\sum_{k\ge1}\frac{u^k}{2k(2k-1)}.
 \label{eq:entropy-series}
\end{equation}
Since \(h'(t)=-\operatorname{arctanh}t\), integrating the series for
\(\operatorname{arctanh}t\) from zero to \(\sqrt u<1\) gives the equality.
Continuity and monotone convergence extend it to \(u=1\), where the sum is
\(\log 2\). Termwise \(u^k\le u\) gives \(g(u)\le u\log 2\), strictly for
\(0<u<1\). Likewise, \((a+b)^k\ge a^k+b^k\) gives
\begin{equation}
 g(a+b)\ge g(a)+g(b)\qquad(a,b\ge0,\ a+b\le1),
 \label{eq:entropy-superadditivity}
\end{equation}
strictly for \(ab>0\), by the \(k=2\) term.

\section{A counterexample in every pair of dimensions at least three}
\label{sec:universal}

\begin{theorem}[Universal rank-two family]\label{thm:universal}
For every pair of integers \(m,n\ge3\), CQC is false on
\(\mathbb C^m\otimes\mathbb C^n\). Let
\begin{equation}
\begin{gathered}
 |u_d\rangle=\frac{|1\rangle+|2\rangle}{\sqrt2}\quad(d\ge3),
 \\
 \rho_{mn}=\tfrac12|00\rangle\langle00|
 +\tfrac12|u_m,u_n\rangle\langle u_m,u_n|,
\end{gathered}
 \label{eq:universalstate}
\end{equation}
and take the computational bases as \(Q_A,Q_B\) and the canonical Fourier
bases as \(R_A,R_B\). Then \(\rho_{mn}\) is separable, has rank two, and
\begin{equation}
 \cmiq+\cmir-\qmi_{\rho_{mn}}\ \geq\ \frac1{8m^2n^2}\ >\ 0 .
 \label{eq:universalexcess}
\end{equation}
\end{theorem}

Theorems~\ref{thm:universal}, \ref{thm:qubittail}
and~\ref{thm:strengthened} are proved in Appendix~\ref{app:proofs}. Here we
give the statements and the mechanisms.

The mechanism is transparent. The two product rays are locally orthogonal on
both sides, so the nonzero spectra of \(\rho_{AB},\rho_A,\rho_B\) are all
\((1/2,1/2)\) and \(\qmi=\log2\). The computational measurement reveals the
latent binary label exactly on each side, so \(\cmiq=\log2\) already saturates
the bound; the Fourier measurement then contributes a strictly positive
remnant of the same label, and the excess is that remnant.

\begin{corollary}[Full-rank separable violations]\label{cor:fullrank}
For every fixed \(m,n\ge3\) there are full-rank separable counterexamples on
\(\mathbb C^m\otimes\mathbb C^n\), obtained by mixing \(\rho_{mn}\) with
sufficiently little maximally mixed noise.
\end{corollary}

\begin{remark}[Relation to the independent refutation]\label{rem:wangcompare}
Wang, Wang and Chen~\cite{Wang2026cx} independently refuted CQC for every
equal local dimension \(d\ge3\), posting first. Their state
\(\tfrac12|u,v\rangle\langle u,v|+\tfrac12|v,u\rangle\langle v,u|\), with
\(|u\rangle\) the uniform Fourier vector and
\(|v\rangle=(|0\rangle-|1\rangle)/\sqrt2\), swaps one pair of local vectors
across the parties, whereas \eqref{eq:universalstate} repeats the same branch
on both sides; the two constructions are different points of the same
rank-two separable classical--classical class, found by different routes.
Where the scopes overlap their violation is larger --- \(0.0484\) versus
\(0.0089\) bits at \(3\times3\), both computed exactly --- and grows with
\(d\), while the excess of \eqref{eq:universalexcess} plateaus. What
\cite{Wang2026cx} does not contain is unequal dimensions:
Theorem~\ref{thm:universal} covers every pair \(m,n\ge3\), and
Theorems~\ref{thm:qubittail} and~\ref{thm:strengthened} the qubit--qudit
tail, neither of which follows from the symmetric construction.
\end{remark}

\section{Qubit--qudit families}\label{sec:qubittail}

Theorem~\ref{thm:universal} needs three computational basis vectors on each
side and therefore says nothing when one party is a qubit. A different
two-ray mixture covers that case.

\begin{theorem}[All-MUB qubit--qudit tail]\label{thm:qubittail}
Let \(n\ge164\), let \((Q_B,R_B)\) be any ordered pair of rank-one projective
mutually unbiased bases of \(\mathbb C^n\), and choose \(q_0\in Q_B\),
\(r_0\in R_B\). On Alice use the Pauli \(Z,X\) bases and set
\begin{equation}
 \omega_n=\tfrac12|0,q_0\rangle\langle0,q_0|
 +\tfrac12|+,r_0\rangle\langle+,r_0|,
 \qquad|+\rangle=\frac{|0\rangle+|1\rangle}{\sqrt2}.
 \label{eq:omega}
\end{equation}
Then \(\omega_n\) is separable of rank two and strictly violates CQC. The
same holds on \(\mathbb C^n\otimes\mathbb C^2\) by exchanging the parties.
\end{theorem}

The uniformity in the MUB pair is what forces \(n\ge164\). To refute CQC in a
\emph{fixed} dimension it suffices to exhibit one legal MUB pair, and that
lowers the boundary sharply.

\section{A near-complete dimension classification}\label{sec:classification}

\begin{theorem}[Strengthened qubit--qudit counterexamples]
\label{thm:strengthened}
For every integer \(n\ge3\) except \(n=3,5\) there is an ordered rank-one
projective MUB pair on \(\mathbb C^n\) and a rank-two separable state on
\(\mathbb C^2\otimes\mathbb C^n\) that strictly violates CQC; the same holds
on \(\mathbb C^n\otimes\mathbb C^2\). The proof combines four explicit
constructions: a computational/Fourier two-spoke family for every \(n\ge12\);
a subgroup/coset family in every composite \(n\ge4\); and exact
complex-Hadamard witnesses at \(n=7\) (Petrescu~\cite{Petrescu1997}) and
\(n=11\) (Nicoar\u a's \(N_{11}\); see the
catalogue~\cite{TadejZyczkowski2006}).
\end{theorem}

\begin{corollary}\label{cor:strengthenedfullrank}
For every fixed \(n\ge3\) with \(n\notin\{3,5\}\), both \(2\times n\) and
\(n\times2\) admit full-rank separable counterexamples.
\end{corollary}

\subsection{A full-rank counterexample at \texorpdfstring{$2\times5$}{2x5}}
\label{sec:n5-witness}

Use the computational and Fourier bases on each subsystem, with the tensor
order \(\mathbb C^2\otimes\mathbb C^5\). Let \(M\) be the exact complex rational
lower triangular factor supplied with the replay instructions in
\texttt{paper/reproduction/README.md}, and put
\begin{equation}
 \rho_5=\frac{MM^\dagger}{\Tr(MM^\dagger)}.
\end{equation}
All diagonal entries of \(M\) are positive, so \(\rho_5\) has trace one and
rank ten. Its measured entropies and quantum mutual information were bounded
using rigorous ball arithmetic. The resulting enclosure gives
\begin{equation}
 \cmiq+\cmir-\qmi_{\rho_5}>2\times10^{-4}\ \text{nats}.
 \label{eq:n5-certified-excess}
\end{equation}
The check uses the fixed rational factor, not the floating-point state from
which it was obtained. No separability claim is made for \(\rho_5\).

Together with Theorems~\ref{thm:universal} and~\ref{thm:strengthened}, this
excludes CQC in every nontrivial dimension pair other than \(2\times2\),
\(2\times3\), and \(3\times2\). The two-qubit case is valid by
Ref.~\cite{Wang2026}. Thus, up to exchange of the parties, only \(2\times3\)
remains open in this classification. If either local dimension is one, all
three mutual informations vanish.

\begin{remark}\label{rem:separable}
The rank-two counterexamples above are separable, so entanglement is not
necessary for CQC failure. The universal family \(\rho_{mn}\) is moreover exactly
classical--classical --- its two branch product rays are orthogonal on both
sides, so it is diagonal in a local product basis --- and the same is true of
the rank-three qutrit witness \eqref{eq:witness}; for those states discord is not
necessary either. The qubit--qudit witnesses \eqref{eq:omega} are separable
but \emph{not} classical--classical: their branch rays overlap on both sides.
This distinction will matter twice below: the classical--classical witnesses
are what make the corrected bound of Section~\ref{sec:corrected} the right
repair, and their alignment with a product basis is one of the reasons random
sampling could not find them (Section~\ref{sec:sampling}).
\end{remark}

\section{One unbiased qubit readout suffices}\label{sec:unbiased-qubit}

Only Alice needs to be a qubit. Her state need not be maximally mixed:
it is enough that one of her two measurements gives a fair bit.
Unlike the two-qubit proof of Ref.~\cite{Wang2026}, this argument puts no
restriction on Bob's dimension or choice of measurements.

\begin{theorem}[One unbiased qubit readout]\label{thm:unbiased-qubit}
Let \(\rho_{AB}\) be a state on \(\mathbb C^2\otimes\mathcal H_B\), with
\(\mathcal H_B\) finite dimensional. Let \(Z,X\) be mutually unbiased
rank-one projective measurements on \(A\), and let \(Q,R\) be arbitrary
finite-outcome POVMs on \(B\). If the marginal outcomes of either \(Z\) or
\(X\) have probabilities \((1/2,1/2)\), then
\begin{equation}
 I(Z:Q)+I(X:R)\le I(A:B)_\rho.
\end{equation}
\end{theorem}

\begin{proof}
Take \(X,Z\) to be Pauli measurements by a local change of basis. Assume
\(Z\) is unbiased. The goal is to show that Alice's \(X\) measurement
loses at least \(I(Z:Q)\) of the original correlation. The remaining
correlation, \(I(X:B)\), already bounds \(I(X:R)\).

\smallskip\noindent\textit{1. A bound for qubit ensembles.}
Consider any ensemble of qubit states with classical label \(L\) and
mean Bloch vector with zero \(z\) coordinate. For the function \(g\) in
\eqref{eq:entropy-series}, superadditivity~\eqref{eq:entropy-superadditivity}
with \(a=x^2+y^2\) and \(b=z^2\) gives
\[
 f(x,y,z):=g(x^2+y^2+z^2)-g(x^2)-g(z^2)\ge f(x,y,0).
\]
On the plane \(z=0\), \(f(x,y,0)=D(\tau\Vert\Delta_X\tau)\), where
\(\tau\) is the qubit with Bloch vector \((x,y,0)\), \(D\) denotes relative
entropy, and \(\Delta_X\) is pinching in the \(X\) basis. Indeed,
\(D(\tau\Vert\Delta_X\tau)=S(\Delta_X\tau)-S(\tau)\).
Joint convexity of \(D\) makes this function convex on the projected disk.
Let \(p_\ell,v_\ell\) be the weights and Bloch vectors, \(\bar v=\sum_\ell p_\ell v_\ell\),
and let \(P(x,y,z)=(x,y,0)\). Since \(P\bar v=\bar v\), the pointwise
bound and convexity on the disk give
\[
 \sum_\ell p_\ell f(v_\ell)-f(\bar v)
 \ge \sum_\ell p_\ell f(Pv_\ell)-f(\bar v)\ge0.
\]
Write \(S(v)\) for the entropy of the qubit with vector \(v\), and
\(H_M(v)\) for the entropy of its outcomes in basis \(M=X,Z\).
Then \(f(v)=-S(v)+H_X(v)+H_Z(v)-\log 2\). Thus the left-hand side
is the ensemble's Holevo information \(S(\bar v)-\sum_\ell p_\ell S(v_\ell)\)
minus \(I(L:X)+I(L:Z)\), since \(I(L:M)=H_M(\bar v)-\sum_\ell p_\ell H_M(v_\ell)\).
\smallskip\noindent\textit{2. Apply the ensemble bound after Bob measures.}
Condition Alice's state on Bob's \(Q\) outcome. Write \(p_q\) for
its probability and \(\rho_A^q\) for the conditional state, omitting zero
probabilities. The ensemble mean is \(\sum_q p_q\rho_A^q=\rho_A\), and the
block entropy of the classical--quantum output is
\(S(AQ)=H(p)+\sum_q p_q S(\rho_A^q)\). Thus its Holevo information is
\(I(A:Q)\). The joint outcome table for Alice's measurement \(M=X,Z\)
has entries \(p_q\Tr(M_m\rho_A^q)\) and is exactly the Born table for
\(M,Q\). The ensemble inequality therefore gives
\(I(A:Q)-I(X:Q)\ge I(Z:Q)\).

\smallskip\noindent\textit{3. Compare the loss before and after Bob measures.}
Data processing for relative entropy under
Bob's \(Q\) measurement, which commutes with Alice's pinching, gives
\[
 I(A:B)-I(X:B)\ge I(A:Q)-I(X:Q).
\]
Indeed, for \(Y=B,Q\) write \(\tau_{AY}\) for the joint state before or
 after Bob's measurement. The pinched state represents the classical
 register \(X\) and leaves the \(Y\) marginal unchanged. Expanding the
 mutual informations thus gives
\[
\begin{aligned}
 I(A:Y)-I(X:Y)
 &= S(\Delta_X\tau_{AY})-S(\tau_{AY}) \\
 &\quad -\bigl[S(\Delta_X\rho_A)-S(\rho_A)\bigr] \\
 &= D(\tau_{AY}\Vert\Delta_X\tau_{AY}) \\
 &\quad-D(\rho_A\Vert\Delta_X\rho_A).
\end{aligned}
\]
The pinching identity follows by tracing \(\log\Delta_X\tau\) against the
block-diagonal part of \(\tau\); the support inclusion holds also for
singular states, since pinching averages unitary conjugates including the
identity. The same measurement channel takes the two joint states for
\(Y=B\) to the two for \(Y=Q\), and the marginal term is unchanged.
Finally, \(I(X:B)\ge I(X:R)\) by data
processing on Bob. Hence
\[
 I(A:B)\ge I(X:B)+I(Z:Q)\ge I(X:R)+I(Z:Q).
\]
For an unbiased \(X\) readout, interchange \((X,R)\) and \((Z,Q)\) throughout.
\end{proof}

In particular, a qubit--qudit violation requires that \emph{both} qubit
readouts be biased. For example, the Bloch vector \((1/3,1/3,0)\) gives a
fair \(Z\) outcome even though Alice's marginal is neither maximally mixed
nor diagonal in either measured basis~\cite{Schneeloch2014}.

\section{What we can show at \texorpdfstring{$2\times3$}{2x3}}
\label{sec:residual}

We show that CQC holds at \(2\times3\) for equal mixtures of two product
states of the form \eqref{eq:family}. Bob's two branch states are orthogonal;
Alice's are \(|0\rangle\) and \(|+\rangle\). The gap exceeds \(0.012\) nats
(Corollary~\ref{cor:closure}). This does not cover general \(2\times3\) states.

First, here is why the simplest counterexample construction does not work
at \(2\times3\). It needs Bob to identify the branch perfectly with either
of his two measurements. This is impossible when his dimension \(d\) is
three or five. The full-rank \(2\times5\) witness does not use this mechanism.

\begin{theorem}[Full-spark barrier]\label{thm:fullspark}
Let \(d\in\{3,5\}\). Every complex Hadamard matrix of order \(d\) is monomial
equivalent to \(F_d\)~\cite{Haagerup1996}, so the computational/Fourier pair
is a normal form for Bob's MUB pair. Moreover \(F_d\) is full
spark~\cite{Tao2005}, equivalently every nonzero
\(\psi\in\mathbb C^d\) obeys the support uncertainty relation
\begin{equation}
 |\operatorname{supp}_Q\psi|+|\operatorname{supp}_R\psi|\ \ge\ d+1 .
 \label{eq:spark}
\end{equation}
Consequently no two nonzero branch states, pure or mixed, can be perfectly
distinguished by both the \(Q\) and the \(R\) measurement.
\end{theorem}

\begin{proof}[Proof sketch]
If branch states \(\tau_0,\tau_1\) were perfectly distinguished by both, then
choosing nonzero \(u\in\operatorname{ran}\tau_0\),
\(v\in\operatorname{ran}\tau_1\) gives disjoint \(Q\)-supports and disjoint
\(R\)-supports, hence \(s_Q(u)+s_Q(v)\le d\) and \(s_R(u)+s_R(v)\le d\).
Adding \eqref{eq:spark} for \(u\) and for \(v\) gives
\(s_Q(u)+s_R(u)+s_Q(v)+s_R(v)\ge2d+2\), contradicting the upper bound
\(2d\).
\end{proof}

Ruling out perfect readout is not enough: partial information from both
measurements could still violate CQC. We now bound how much information the
two measurements can reveal in the stated family.

\subsection{The equal-prior orthogonal two-ray family}\label{sec:family}

Fix Alice's ordered MUB pair \((Z,X)\), let Bob have dimension \(d\), and let
\(u\perp v\) be orthonormal in \(\mathbb C^d\). Put
\begin{equation}
 \rho=\tfrac12|0,u\rangle\!\langle0,u|+\tfrac12|+,v\rangle\!\langle+,v| .
 \label{eq:family}
\end{equation}
The branch rays are locally orthogonal on Bob and have Alice overlap
\(2^{-1/2}\), so \(\qmi=h(1/\sqrt2)\). Write \(p^Q_j=|\langle q_j,u\rangle|^2\),
\(q^Q_j=|\langle q_j,v\rangle|^2\), likewise in \(R\), and let
\begin{equation}
\begin{gathered}
 \Delta(p,q)=\sum_j\frac{(p_j-q_j)^2}{p_j+q_j},
 \\
 S=\Delta(p^Q,q^Q)+\Delta(p^R,q^R)\in[0,4]
\end{gathered}
 \label{eq:Sdef}
\end{equation}
be the triangular-discrimination sum, omitting any term with \(p_j+q_j=0\).
Alongside it we use the total
variation and Bhattacharyya functionals
\(\TV(p,q)=\frac12\sum_j|p_j-q_j|\) and
\(\mathrm{BC}(p,q)=\sum_j\sqrt{p_jq_j}\), abbreviated
\(\TV_Q=\TV(p^Q,q^Q)\) and likewise \(\TV_R,\mathrm{BC}_Q,\mathrm{BC}_R\).
(In this section the bare letter \(S\) always denotes the
triangular-discrimination sum; von Neumann entropy appears only with a
subsystem argument.)

The perfect-double-readout witnesses of Appendix~\ref{app:proofs} belong
to this family in the dimensions where that mechanism is available. In
contrast, the all-MUB witness \eqref{eq:omega} has nonorthogonal Bob rays,
and the full-rank \(2\times5\) witness is outside the family by rank.
Orthogonality on Bob fixes \(\qmi=h(1/\sqrt2)\), so it suffices to control
the measured correlations through the single scalar \(S\).

The \((Z,Q)\) table is
\(\Pr[Z=0,Q=j]=p^Q_j/2+q^Q_j/4\), \(\Pr[Z=1,Q=j]=q^Q_j/4\), with row marginal
\((3/4,1/4)\). Setting \(m_j=(p^Q_j+q^Q_j)/2\) and
\(x_j=(p^Q_j-q^Q_j)/(p^Q_j+q^Q_j)\) for \(m_j>0\) gives the exact reduction
(with zero-probability outcomes omitted)
\begin{equation}
\begin{gathered}
 I(Z:Q)=\sum_j m_j\,\delta(x_j),
 \\
 \delta(x)=\frac{3+x}4\log\frac{3+x}3+\frac{1-x}4\log(1-x),
\end{gathered}
 \label{eq:table}
\end{equation}
together with the three moment identities
\begin{equation}
\begin{gathered}
 \sum_j m_j=1,\qquad\sum_j m_jx_j=0,\\
 \sum_j m_jx_j^2=\tfrac12\Delta(p^Q,q^Q).
\end{gathered}
 \label{eq:moments}
\end{equation}
The \((X,R)\) table is the same with the branches exchanged, so
\eqref{eq:table}--\eqref{eq:moments} hold verbatim there.

\subsection{A dimension-free entropy envelope}\label{sec:envelope}

Let \(\ell\) be the chord of \(\delta\) on \([-1,1]\). Since
\(\delta(-1)=\frac12\log\frac43\) and \(\delta(1)=\log\frac43\),
\begin{equation}
 \ell(x)=\frac{\log(4/3)}4(3+x).
\end{equation}

\begin{lemma}[Envelope]\label{lem:envelope}
For every \(x\in[-1,1]\),
\begin{equation}
 \delta(x)\ \le\ \ell(x)-\frac{\log2}4\left(1-x^2\right).
 \label{eq:envelope}
\end{equation}
\end{lemma}

\begin{proof}
Put \(x=1-4u\), so \(0\le u\le1/2\). The identities
\(\ell(x)-\delta(x)=H_2(u)-u\log4\) and \(1-x^2=8u(1-2u)\) reduce the
claim to \(H_2(u)\ge4u(1-u)\log2\). Taking \(t=1-2u\), this follows from
\(g(t^2)\le t^2\log2\) by \eqref{eq:entropy-series}. Equality holds only
at \(u=0,1/2\), that is, \(x=1,-1\).
\end{proof}

\begin{theorem}[Reduction to one scalar]\label{thm:reduction}
For the family \eqref{eq:family}, in \emph{every} Bob dimension,
\begin{equation}
 G\ \ge\ h(1/\sqrt2)-\frac32\log\frac43
 +\frac{\log2}4\left(2-\frac S2\right).
 \label{eq:Gbound}
\end{equation}
Consequently \(G>0\) whenever
\begin{equation}
 S<S_\star:=4-\frac8{\log2}\left(\frac32\log\frac43-h(1/\sqrt2)\right)
 \approx 3.83
 \label{eq:Sstar}
\end{equation}
\end{theorem}

\begin{proof}
Average \eqref{eq:envelope} against \(m\). The chord is affine and
\(\sum_jm_jx_j=0\), so \(\sum_jm_j\ell(x_j)=\ell(0)=\frac34\log\frac43\),
while \(\sum_jm_jx_j^2=\frac12\Delta\). This gives
\(I(Z:Q)\le\frac34\log\frac43-\frac{\log2}4\bigl(1-\Delta(p^Q,q^Q)/2\bigr)\),
and likewise for \(I(X:R)\); summing and subtracting from
\(\qmi=h(1/\sqrt2)\) gives \eqref{eq:Gbound}. Solving for the zero of this
bound gives \eqref{eq:Sstar}.
\end{proof}

\begin{remark}\label{rem:whyenvelope}
The obvious majorant is the quadratic one, \(\delta(x)\le\log(4/3)\,x^2\), which
gives closure only under \(S\le8/3\). Section~\ref{sec:n5} proves that
\(S=(14+2\sqrt5)/5=3.694427\ldots>8/3\) is \emph{attained} at \(d=5\), so the
quadratic route is unconditionally dead there. Lemma~\ref{lem:envelope} is
what brings \(d=5\) back into range at all: it demands only
\(S<S_\star\).
\end{remark}

\subsection{Bounding the two readouts from Bob's two-dimensional span}\label{sec:gate}

We want an upper bound on \(S\) that works for every orthogonal pair
\(u,v\). First fix the plane they span. Within that plane, maximizing
\(S\) reduces to finding the largest eigenvalue of a \(3\times3\) matrix.

Let \(u\perp v\) span the two-dimensional \(W\subset\mathbb C^d\) and put
\(J=P_u-P_v\). Since
\(|\langle e,u\rangle|^2+|\langle e,v\rangle|^2=\langle e|P_W|e\rangle\),
\begin{equation}
 S(u,v)=\sum_{e\in Q\cup R}\frac{\langle e|J|e\rangle^2}{\langle e|P_W|e\rangle}.
\end{equation}
Omit terms with zero projected norm in all ratios and overlap tables below.
Every traceless Hermitian \(J\) supported on the two-dimensional \(W\) with
\(\Tr J^2=2\) has eigenvalues exactly \(\pm1\), hence \emph{is} some
\(P_u-P_v\). (The letter \(X\) is reserved for the scalar cross block
introduced in \eqref{eq:split} below.) Therefore, with \(A_W\) the Gram matrix of that quadratic form on
the real three-dimensional space of traceless Hermitians supported on \(W\),
taken with the inner product \(\Tr(XY)\),
\begin{equation}
 \max_{u,v\ \mathrm{spanning}\ W}S(u,v)=2\lmax(A_W)
 \label{eq:lmax}
\end{equation}
\emph{exactly}, in every dimension. No completing vector is needed.

Write \(g_e=\langle e|P_W|e\rangle\) and
\(N_{ee'}=|\langle e|P_W|e'\rangle|^2\).

\begin{lemma}[Trace identity]\label{lem:trace}
\(\sum_eg_e=4\), \(\sum_{e,e'}N_{ee'}=8\), and \(\Tr A_W=2\) for every
two-dimensional \(W\) in every dimension.
\end{lemma}

\begin{proof}
The first two follow from \(\sum_e|e\rangle\!\langle e|=2I\) and
\(\Tr P_W=2\). Complete a trace-orthonormal basis \((T_i)_{i=1}^3\) of the
traceless Hermitians on \(W\) by adding \(P_W/\sqrt2\). Expanding the trace
inner product of \(P_W|e\rangle\langle e|P_W\) and
\(P_W|e'\rangle\langle e'|P_W\) in this basis gives
\begin{equation}
 \sum_i\langle e|T_i|e\rangle\langle e'|T_i|e'\rangle
 =N_{ee'}-g_eg_{e'}/2 .
 \label{eq:polarized-bloch}
\end{equation}
At \(e=e'\) the right side is \(g_e^2/2\), so
\(\Tr A_W=\frac12\sum_e g_e=2\).
\end{proof}

Lemma~\ref{lem:trace} has two consequences that do the real work.

\begin{corollary}[The gate is one determinant]\label{cor:det}
Order the eigenvalues of \(A_W\) as \(\lambda_1\ge\lambda_2\ge\lambda_3\ge0\).
Since \(\Tr A_W=2\), we have \(2\lambda_2\le\lambda_1+\lambda_2\le2\).
Thus \(4/3-\lambda_2\) and \(4/3-\lambda_3\) are positive, and
\begin{equation}
 \lmax(A_W)\le\tfrac43\iff\det\!\left(\tfrac43I-A_W\right)\ge0 .
\end{equation}
\end{corollary}

\begin{theorem}[Eigenvalue-free form]\label{thm:e2}
With \(c_{ee'}=N_{ee'}/(g_eg_{e'})\in[0,1]\) the squared cosine between the
projections \(P_We\) and \(P_We'\),
\begin{equation}
\begin{gathered}
 \Tr(A_W^2)=\sum_{e,e'}\frac{N_{ee'}^2}{g_eg_{e'}}-4,
 \\
 e_2(A_W)=\frac12\sum_{e,e'}g_eg_{e'}c_{ee'}(1-c_{ee'})\ \ge\ 0,
\end{gathered}
\end{equation}
and \(\Tr(A_W^2)=4-2e_2(A_W)\). To obtain the first identity, square
\eqref{eq:polarized-bloch}, divide by \(g_eg_{e'}\), and sum over \(e,e'\).
The other two follow from \(2e_2(A_W)=(\Tr A_W)^2-\Tr(A_W^2)\) and
Lemma~\ref{lem:trace}. Since
\(\lmax^2\le\Tr(A_W^2)\), the whole
programme reduces to the single eigenvalue-free inequality
\begin{equation}
 e_2(A_W)\ \ge\ 2-\frac{(S_\star/2)^2}2
 \ \Longrightarrow\ S\le S_\star .
 \label{eq:e2floor}
\end{equation}
\end{theorem}

\subsection{How overlap limits the two readouts}\label{sec:variance}

Split \(A_W=A_Q+A_R\) and let \(D_Q=2-2\Tr(A_Q^2)\), \(D_R\) likewise.
Applying the Bloch identity of Lemma~\ref{lem:trace} one basis at a time
gives \(\Tr A_Q=\frac12\sum_{e\in Q}g_e=1=\Tr A_R\), and
\(\Tr(M^2)\le(\Tr M)^2\) for positive semidefinite \(M\), so
\(D_Q,D_R\ge0\). Write
\(\nu_j=g_{q_j}/2\), \(\nu'_k=g_{r_k}/2\); both are probability vectors by
Lemma~\ref{lem:trace}. Then, exactly,
\begin{equation}
 e_2(A_W)=\frac{D_Q+D_R}4+X,
 \qquad
 X=\sum_{j,k}\nu_j\nu'_k\,4c_{jk}(1-c_{jk}).
 \label{eq:split}
\end{equation}

\begin{theorem}[Variance form]\label{thm:variance}
Completeness of the two bases alone --- not unbiasedness --- pins every row
and column average of the overlap table:
\begin{equation}
\begin{gathered}
 \sum_k\nu'_kc_{jk}=\tfrac12\ \ \forall j,
 \qquad
 \sum_j\nu_jc_{jk}=\tfrac12\ \ \forall k,
 \\
 \mathbb E_{\nu\otimes\nu'}[c]=\tfrac12 .
\end{gathered}
\end{equation}
Therefore
\begin{equation}
 X=1-4\operatorname{Var}_{\nu\otimes\nu'}(c),
\end{equation}
and since \(c\in[0,1]\) has mean \(1/2\),
\begin{equation}
 \operatorname{Var}(c)\le\tfrac14,
 \qquad\text{with equality iff every }c_{jk}\in\{0,1\},
\end{equation}
that is, iff every projected \(Q\) vector is parallel or orthogonal to every
projected \(R\) vector --- precisely perfect double readout. Consequently the
family closes in dimension \(d\) as soon as
\begin{equation}
 \operatorname{Var}_{\nu\otimes\nu'}(c)\le\frac{S_\star^2/8-1}4
 \approx0.20758
 \label{eq:varcap}
\end{equation}
\end{theorem}

\begin{proof}
For the row identity, \(\sum_kg_{r_k}|\hat v'_k\rangle\!\langle\hat v'_k|=P_W\)
resolves the identity on \(W\) and \(\hat v_j\in W\), so
\(\sum_k2\nu'_kc_{jk}=\langle\hat v_j|P_W|\hat v_j\rangle=1\). The global mean
is its \(\nu\)-average. Then
\(X=4\mathbb E[c]-4\mathbb E[c^2]=2-4(\operatorname{Var}(c)+\frac14)
 =1-4\operatorname{Var}(c)\), and
\(\operatorname{Var}(c)=\mathbb E[c^2]-\frac14\le\mathbb E[c]-\frac14=\frac14\)
with equality iff \(\mathbb E[c(1-c)]=0\).
\end{proof}

Theorem~\ref{thm:variance} is the quantitative form of
Theorem~\ref{thm:fullspark}: the qualitative statement ``perfect double
readout is impossible'' becomes ``the variance is bounded away from its
maximum \(1/4\)'', and \eqref{eq:varcap} says exactly how far is enough.
Numerical candidates for the maximum variance at \(d=3,5\) are \(1/8\) and
\(0.190423\), respectively, both below the sufficient cap. These are not
proved upper bounds and concern only the two-ray family. In particular,
they do not constrain the full-rank \(2\times5\) witness.

\subsection{The interval bound at \texorpdfstring{$2\times3$}{2x3}}\label{sec:closure}

At \(d=3\) every two-dimensional \(W\) is \(w^\perp\) for a unique \(w\) up to
phase, so \(X\) is a function of four real parameters. \(X\) is invariant
under a global phase of \(w\), and we parameterise
\(w=(\sqrt{s_0},\sqrt{s_1}e^{i\phi_1},\sqrt{s_2}e^{i\phi_2})\) with
\((s_0,s_1)\) on the two-simplex.

\begin{theorem}[Rigorous cover]\label{thm:cover}
For every unit \(w\in\mathbb C^3\),
\begin{equation}
 X(w)\ge 0.30.
\end{equation}
\end{theorem}

The replay instructions are in \texttt{paper/reproduction/README.md}.
A branch-and-bound calculation in \(128\)-bit interval arithmetic covers
the full four-parameter domain, with no symmetry reduction or undecided
boxes. When a denominator \(g_e\) can vanish in a box, the corresponding
term is bounded below by zero: Cauchy--Schwarz on
\(I-|w\rangle\!\langle w|\) bounds it in \([0,\nu_j\nu'_k]\).

\begin{corollary}[\(2\times3\) closure for the family]\label{cor:closure}
Every state of the form \eqref{eq:family} on
\(\mathbb C^2\otimes\mathbb C^3\), with the computational/Fourier MUB pair on
Bob, satisfies
\begin{equation}
 G\ge 0.012\ \text{nats}.
\end{equation}
\end{corollary}

\begin{proof}
By \eqref{eq:split} and \(D_Q,D_R\ge0\) we get \(e_2\ge X\ge0.30\); by
Theorem~\ref{thm:e2}, \(\Tr(A_W^2)=4-2e_2\le3.4\), hence
\(\lmax\le\sqrt{17/5}\). Substituting \(S\le2\sqrt{17/5}\) into
\eqref{eq:Gbound} gives a lower bound strictly greater than \(0.012\) nats.
\end{proof}

\begin{remark}[Scope]\label{rem:scope}
Corollary~\ref{cor:closure} closes one family, not the dimension pair. The
family is determined by the pair of projectors \((P_u,P_v)\) with
\(u\perp v\), which is \(4+2=6\) real parameters, inside the
\(6^2-1=35\)-parameter state space of \(\mathbb C^2\otimes\mathbb C^3\).
Unequal priors, non-orthogonal Bob rays, Alice branch rays other than
\(|0\rangle,|+\rangle\), higher rank and entangled states are all untouched.
Two restrictions are, however, harmless: Alice's MUB pair may be fixed to
\((Z,X)\) by a local unitary, and at \(d=3\) every complex Hadamard is monomial
equivalent to \(F_3\), so the computational/Fourier choice on Bob is a normal
form.
\end{remark}

\begin{remark}[The relaxation has a ceiling]
The Frobenius step \(\lmax\le\sqrt{\Tr A_W^2}\) can never recover the sharp
qutrit constant from \(X\) alone: even a lower bound \(X\ge1/2\) would give
only \(S\le2\sqrt3>8/3\) by this argument. The determinant test of
Corollary~\ref{cor:det} is an alternative. The sharp constant remains open.
\end{remark}

Appendix~\ref{app:two-ray-details} gives a qutrit equality example and
explains why several proposed bounds fail for this family at \(2\times5\).
Those failures do not affect the full-rank counterexample of
Sec.~\ref{sec:n5-witness}.

\section{A corrected universal inequality}\label{sec:corrected}

How large can the sum of the two measured mutual informations be? We use
the uncertainty relation with quantum memory to bound it. First keep Bob's
system quantum; then subtract the information lost when Bob measures it.
The same argument with the parties exchanged gives a second bound, and
we take the smaller one.

The local dimensions may differ, and the correction must use each party's own
dimension: write \(d_A=\dim\mathcal H_A\) and \(d_B=\dim\mathcal H_B\).
Define the one-sided measurement losses
\begin{equation}
\begin{gathered}
 D_B^Q=I(Q_A:B)-\cmiq,\\ D_B^R=I(R_A:B)-\cmir,
\end{gathered}
\end{equation}
and \(D_A^Q,D_A^R\) by interchanging the parties. Set
\begin{equation}
\begin{gathered}
 C_A^{(0)}=H(Q_A)+H(R_A)-\log d_A-S(A),
 \\
 C_B^{(0)}=H(Q_B)+H(R_B)-\log d_B-S(B).
\end{gathered}
\end{equation}

\begin{theorem}[Corrected bound]\label{thm:corrected}
With \(C_A=C_A^{(0)}-D_B^Q-D_B^R\) and \(C_B=C_B^{(0)}-D_A^Q-D_A^R\),
\begin{equation}
 \cmiq+\cmir \leq \qmi+\min\{C_A,C_B\}
 \label{eq:corrected}
\end{equation}
holds for every state and every pair of local MUB measurements. Dropping the
nonnegative losses gives the simpler marginal form
\(\cmiq+\cmir\le\qmi+\min\{C_A^{(0)},C_B^{(0)}\}\).
\end{theorem}

\begin{proof}
The MUB entropic uncertainty relation with quantum memory~\cite{Berta2010},
applied on party \(A\), gives \(H(Q_A|B)+H(R_A|B)\ge\log d_A+S(A|B)\): the constant is
\(\log(1/c)\) with \(c=\max_{j,k}|\langle q_j|r_k\rangle|^2=1/d_A\), so it is
\(A\)'s dimension that enters, whatever \(d_B\) is. Hence
\begin{equation}
\begin{aligned}
 I(Q_A:B)+I(R_A:B)
 &=H(Q_A)+H(R_A)\\
 &\quad-H(Q_A|B)-H(R_A|B)\\
 &\le H(Q_A)+H(R_A)
 \\
 &\quad -\bigl[\log d_A+S(A|B)\bigr]\\
 &=\qmi+C_A^{(0)}.
\end{aligned}
\end{equation}
Subtracting the two loss definitions yields the \(A\)-directed version;
interchanging the parties gives the \(B\)-directed one; taking the minimum
proves \eqref{eq:corrected}. The simpler form follows from data processing,
since all four losses are nonnegative.
\end{proof}

Because the memoryless state-dependent MUB relation gives
\(H(Q_A)+H(R_A)\ge\log d_A+S(A)\), the terms \(C_A^{(0)},C_B^{(0)}\) are
nonnegative. Since each outcome entropy on \(A\) is at most \(\log d_A\) and
likewise on \(B\), the bound is never weaker than the coarse correction
\(\qmi+\min\{\log d_A-S(A),\log d_B-S(B)\}\). It also recovers two important
zero-correction cases: if \(\rho_A\) is maximally mixed then
\(C_A^{(0)}=0\), and if \(\rho_A\) is diagonal in \(Q_A\) then \(H(Q_A)=S(A)\)
and MUB measurement forces \(H(R_A)=\log d_A\), again giving \(C_A^{(0)}=0\).
The correction does not vanish automatically for every pure bipartite state,
so it does not subsume every previously known exact CQC class, and it is not
tight on the witnesses of this paper.

The marginal term \(C_A^{(0)}\) is the excess in Alice's memoryless
uncertainty bound, while \(D_B^Q+D_B^R\) is the correlation lost to Bob's
readouts. A nonpositive corrected term on either side implies CQC for
that state.

\begin{remark}[Unequal dimensions]\label{rem:unequal}
The overlap in Alice's uncertainty relation is \(1/d_A\); replacing
\(\log d_A\) by \(\log d_B\) there is not valid in general.
\end{remark}

\label{sec:incomparable}

Plain CQC is true at \(2\times2\)~\cite{Wang2026}, so it is natural to ask
whether the correction degenerates there --- that is, whether
\(\min\{C_A,C_B\}\le0\) for every two-qubit state, which would make
\eqref{eq:corrected} a strengthening of a true statement and the two-qubit
theorem an immediate corollary. It does not. The following two exact
one-parameter families drive the correction to opposite signs, so the two
inequalities cross. Both parties measure the Pauli pair \(Q=Z\), \(R=X\).

\begin{proposition}[\(2\times2\) incomparability]\label{prop:incomparable}
Let \(\tau_r=\tfrac12(\mathbb 1+r\sigma_y)\) for \(r\in(0,1]\), where
\(\sigma_y\) is the Pauli \(Y\) matrix; and let
\(|\psi_\theta\rangle=\cos\theta\,|00\rangle+\sin\theta\,|11\rangle\), with
\(0\le\theta\le\pi/2\).
\begin{enumerate}
\item[(i)] On the product state \(\rho_r=\tau_r\otimes\tau_r\) all four losses
vanish and
\begin{equation}
 \min\{C_A,C_B\}=\log2-H_2\!\left(\tfrac{1+r}{2}\right)>0,
 \label{eq:posfamily}
\end{equation}
while CQC holds with equality. Hence \eqref{eq:corrected} is strictly weaker
than CQC there and does not imply it pointwise.
\item[(ii)] On \(|\psi_\theta\rangle\) one has \(C_A^{(0)}=0\) and
\eqref{eq:corrected} holds with \emph{equality} for every \(\theta\), while
CQC has the strictly positive slack
\(H_2(\cos^2\theta)+H(p^{XX})-2\log2\) for
\(\theta\notin\{0,\pi/4,\pi/2\}\), where
\(p^{XX}\) is the \(X\otimes X\) outcome law. Hence \eqref{eq:corrected} is
strictly stronger there, and it is the tight one.
\end{enumerate}
Consequently the two inequalities are incomparable at \(2\times2\): neither
implies the other, and the two-qubit CQC theorem is not a corollary of
\eqref{eq:corrected}.
\end{proposition}

\begin{proof}
(i) The Bloch vector of \(\tau_r\) points along \(y\), which is unbiased to
both measured bases, so \(H(Q_A)=H(R_A)=\log2\) exactly while
\(S(A)=H_2((1+r)/2)\); this gives
\(C_A^{(0)}=C_B^{(0)}=\log2-H_2((1+r)/2)\), positive for \(r>0\). The state
is a product, so every mutual information appearing in the losses vanishes and
\(C_A=C_A^{(0)}\), which is \eqref{eq:posfamily}. Both sides of CQC are zero.

(ii) The Schmidt basis is the measured \(Q\) basis, so \(H(Q_A)=S(A)\) and
\(H(R_A)=\log2\), whence \(C_A^{(0)}=0\). For the equality, note that
\eqref{eq:corrected} is the Berta relation with the two losses moved across,
so it is saturated exactly when that relation is. Here Bob predicts Alice's
\(Z\) outcome perfectly, so \(H(Q_A|B)=0\); the \(X\) outcomes are uniform
with pure, hence zero-entropy, conditional states on \(B\), so
\(H(R_A|B)=\log2-S(A)\); and \(S(A|B)=-S(A)\) by purity. The sum is
\(\log2+S(A|B)\), which is equality in the Berta relation. The CQC slack is
the stated expression because \(I(A:B)=2H_2(\cos^2\theta)\),
\(\cmiq=H_2(\cos^2\theta)\) and \(\cmir=2\log2-H(p^{XX})\), with
\(p^{XX}=\big((c+s)^2/4,(c-s)^2/4,(c-s)^2/4,(c+s)^2/4\big)\),
\(c=\cos\theta\), \(s=\sin\theta\). Its strict positivity is
Lemma~\ref{lem:slackpos} below.
\end{proof}

\begin{lemma}[Strict CQC slack on the aligned pure family]\label{lem:slackpos}
For \(|\psi_\theta\rangle\) as above, put \(x=|\cos2\theta|\) and
\(y=|\sin2\theta|\). Then
\begin{equation}
 G=h(x)+h(y)-\log2=g(1)-g(x^2)-g(y^2),
 \label{eq:slackseries}
\end{equation}
which is nonnegative and vanishes exactly when \(xy=0\).
\end{lemma}

\begin{proof}
The two parity classes of \(p^{XX}\) have weights \((1\pm\sin2\theta)/2\),
with equal probabilities within each class. Hence \(H(p^{XX})=\log2+h(y)\).
Together with \(H_2(\cos^2\theta)=h(x)\) this gives \eqref{eq:slackseries}
from the slack formula in Proposition~\ref{prop:incomparable}.
Since \(x^2+y^2=1\), \eqref{eq:entropy-superadditivity} gives the sign and
exact equality cases.
\end{proof}

The two examples distinguish marginal uncertainty from Bob's readout
losses. In (i) there are no losses, while in (ii) the marginal term
vanishes. More generally, \(C_A\le C_A^{(0)}\le\log d_A\), and likewise
on Bob. The product example with \(r=1\) attains the resulting ceiling
\(\log 2\) on \(\min\{C_A,C_B\}\) for two qubits.

\begin{remark}[Which bound is stronger is state dependent]\label{rem:census}
A seeded Hilbert--Schmidt census of \(4000\) two-qubit states puts
\(97.2\%\) of them on the side where \(\min\{C_A,C_B\}<0\), so the corrected
bound is typically the stronger statement, with \(2.0\%\) exceeding
\(+0.01\); the observed range was \([-0.614,+0.112]\), and
\eqref{eq:corrected} was never violated. This is a sampled statistic about a
particular ensemble, not a bound; the exact content is
Proposition~\ref{prop:incomparable}.
\end{remark}

\section{Why random sampling did not find these states}\label{sec:sampling}

The conjecture survived \(10^7\) random \(3\otimes3\) samples in the original
work and further random tests in dimensions three and five
in~\cite{Iqbal2026}. Our fixed computational/Fourier simulations illustrate
how these searches can miss selected witnesses. Rotating the measurement
bases together on each party is equivalent to applying the inverse local
rotation to the state; the alignment experiment in
Sec.~\ref{sec:align} probes this effect for the qutrit witness.

\subsection{The gap is concentrated far from zero}\label{sec:concentration}

We sampled \(2\times10^4\) states from each of five standard ensembles on
\(\mathbb C^3\otimes\mathbb C^3\) --- induced measures with ancilla dimension
\(K=1\) (pure), \(3\), \(9\) (Hilbert--Schmidt) and \(18\), and the Bures
measure --- and evaluated \(G\).

\begin{table*}[t]
\caption{CQC gaps for \(2\times10^4\) draws per ensemble.}
\label{tab:sampling-gaps}
\begin{ruledtabular}
\begin{tabular}{lrrrrr}
Ensemble & Mean \(G\) & Std. dev. & Min. \(G\) & \((\text{mean}-\min)/\text{std. dev.}\) & \(\#\{G<0\}\)\\
\hline
Hilbert--Schmidt \((K=9)\) & 0.34615 & 0.04930 & 0.19071 & 3.15 & 0\\
induced \(K=3\)            & 0.81761 & 0.10165 & 0.41781 & 3.93 & 0\\
induced \(K=18\)           & 0.17289 & 0.02884 & 0.07008 & 3.57 & 0\\
induced \(K=1\) (pure)     & 1.00125 & 0.25446 & 0.13981 & 3.39 & 0\\
Bures                      & 0.46578 & 0.06669 & 0.23160 & 3.51 & 0\\
\end{tabular}
\end{ruledtabular}
\end{table*}

In Table~\ref{tab:sampling-gaps}, the means and standard deviations are stable to well under one
per cent across seeds; the minima are extreme order statistics and vary by
several per cent from run to run, so they should be read as indicative of the
scale rather than as reproducible constants. The \(\sigma\)-distances stay in
the range \(3\)--\(4.5\) in every replication we performed; the run tabulated
here spans \(3.15\) to \(3.93\).

No sample in \(10^5\) came close to the boundary \(G=0\). For comparison we
use the rank-three qutrit witness \(\rho_\star\): with
\(|u\rangle=(|0\rangle-|1\rangle)/\sqrt2\), \(|v\rangle=|2\rangle\) and
\(P_x=|x\rangle\!\langle x|\),
\begin{equation}
 \rho_\star=\tfrac13\left(P_u\otimes P_u+P_u\otimes P_v+P_v\otimes P_u\right),
 \label{eq:witness}
\end{equation}
measured in the computational and Fourier bases. It has
\begin{equation}
 G=-\tfrac29\log\tfrac98-\tfrac79\log\tfrac{63}{64}
 =-0.013925285837\ldots,
\end{equation}
which sits \((0.34615+0.01393)/0.04930=7.30\) standard deviations below the
Hilbert--Schmidt mean, whereas the observed minimum of \(2\times10^4\) samples
reached only \(3.15\). Figure~\ref{fig:concentration}(a) shows the five
distributions against the violating half-line.

\begin{figure*}[t]
\centering
\includegraphics[width=0.92\textwidth]{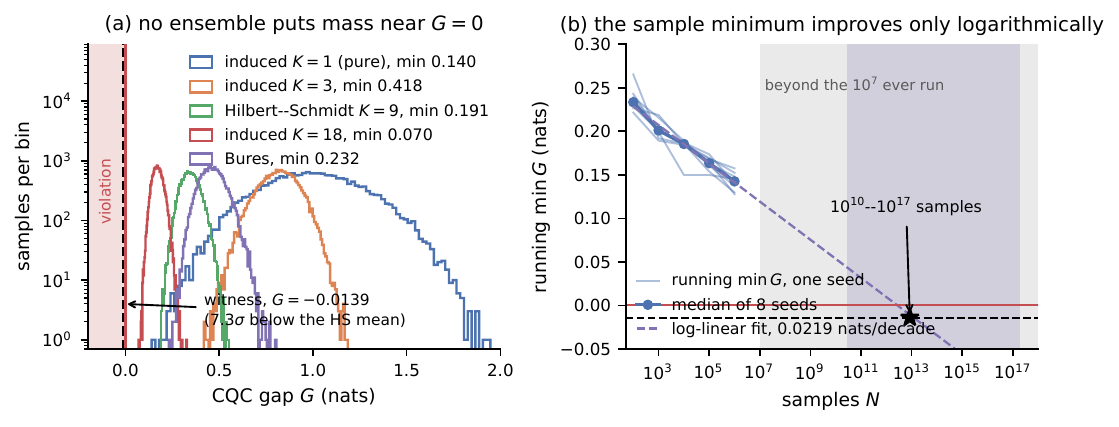}
\caption{The gap is concentrated far from zero, and the sample minimum barely
moves. \textbf{(a)} Distribution of \(G\) over \(2\times10^4\) states from
each of five standard ensembles on \(\mathbb C^3\otimes\mathbb C^3\); the
shaded half-line is the violating region \(G<0\). No sampled gap lies
there. The witness \eqref{eq:witness} (dashed) sits \(7.3\) standard
deviations below the Hilbert--Schmidt mean, while the best of \(2\times10^4\)
samples reaches only \(3.15\)--\(3.93\). \textbf{(b)} Running minimum of \(G\)
over \(10^6\) Hilbert--Schmidt samples for eight independent seeds, with a
log-linear fit to the median trace. The shaded band on the horizontal axis is
the spread, across those seeds, of the sample count at which the
extrapolation would reach the witness value; the grey region marks everything
beyond the \(10^7\) samples ever actually run. The width of that band is the
honest precision of the extrapolation.}
\label{fig:concentration}
\end{figure*}

For the qutrit witness, the full-rank white-noise neighbourhood is narrow:
in \(469\) of \(500\) sampled directions from an interior violating point,
the state-space boundary is reached before the CQC gap becomes
nonnegative. The median alignment tolerance over \(600\) sampled
local-unitary directions is \(0.156\) rad. Appendix~\ref{app:sampling-details}
gives the running-minimum fits, geometry and alignment diagnostics, and
heuristic extrapolations. These experiments describe the tested qutrit
witness and sampling ensembles, not a general bound on detection rates.

\section{Consequences}\label{sec:consequences}

The rank-two counterexamples are separable, and the universal family together with the
qutrit witness \eqref{eq:witness} is exactly classical--classical
(Remark~\ref{rem:separable}), so neither entanglement nor --- for those
states --- discord is necessary for CQC failure. The standard relation \eqref{eq:cqc} cannot be used as a universal
lower bound on quantum mutual information, and any operational claim that uses
it as an essential implication needs an independent replacement argument;
Theorem~\ref{thm:corrected} supplies a universal one at the cost of a
state-and-measurement dependent correction.

The residual analysis also isolates what is special about \(d=3,5\). The rank-two
qubit--qudit mechanism uses perfect double readout, which
Theorem~\ref{thm:fullspark} forbids at prime \(d=3,5\) via full spark, and
Theorem~\ref{thm:variance} shows that this obstruction is quantitatively a
variance identity: \(\operatorname{Var}(c)=1/4\) exactly characterises perfect
double readout. At \(d=3\), the interval cover bounds the variance away from
\(1/4\) by enough to establish CQC within the stated family. At \(d=5\), the
full-spark obstruction excludes perfect double readout but does not settle
CQC even within that family.

\section{Open problems}

\begin{enumerate}
\item \textbf{The two-ray family at \(2\times5\).} Prove an upper bound \(S\le S_\star\) for
 orthonormal pairs in \(\mathbb C^5\). The candidate route is
 \eqref{eq:refined}; the sharp value is conjecturally \((14+2\sqrt5)/5\),
 whose attainment is Theorem~\ref{thm:n5}.
\item \textbf{Beyond the two-ray family at \(2\times3\).} Unequal priors,
 non-orthogonal Bob rays, general Alice rays, and higher rank remain open;
 numerically the gap approaches zero only in degenerate limits, never
 becoming negative.
\item \textbf{The sharp qutrit constant.} Prove \(S\le8/3\) at \(d=3\), which
 by Corollary~\ref{cor:det} is one determinant sign, and which the Frobenius
 relaxation provably cannot reach.
\item \textbf{Smallest total dimension.} Does \(2\times3\) admit a
 counterexample, or is \(3\times3\) the smallest counterexample by total
 Hilbert-space dimension?
\end{enumerate}

\section*{Competing interests}

The authors declare no competing interests.

\section*{Data availability}
The exact input and replay instructions for the \(2\times5\) certificate
are supplied with the manuscript reproduction materials.

\appendix

\section{Proofs of the counterexample theorems}\label{app:proofs}

This appendix proves Theorems~\ref{thm:universal}, \ref{thm:qubittail}
and~\ref{thm:strengthened} and their corollaries. Two steps evaluate an
explicit elementary function of one variable at a single rational point and
use the sign of the result; those evaluations are stated with enough digits
that the reader can repeat them, and in both cases the remaining argument is
analytic.

\subsection{The universal family}

\begin{proof}[Proof of Theorem~\ref{thm:universal}]
Equation~\eqref{eq:universalstate} is a convex combination of two product
projectors, so \(\rho_{mn}\) is a separable density operator. The two product
vectors are orthogonal on each local system. Hence \(\rho_{mn}\),
\((\rho_{mn})_A\) and \((\rho_{mn})_B\) each have nonzero spectrum
\((1/2,1/2)\), and therefore
\begin{equation}
 \qmi_{\rho_{mn}}=\log2. \label{eq:app-qmi}
\end{equation}

Introduce the equiprobable latent label \(L\in\{0,1\}\) for the two terms of
the mixture. Conditional on \(L=0\) each computational outcome is zero;
conditional on \(L=1\) each lies in \(\{1,2\}\). So either computational
outcome determines \(L\), and the two outcomes are conditionally independent
given \(L\). Therefore
\begin{equation}
 \cmiq=H(L)=\log2, \label{eq:app-computational}
\end{equation}
which already saturates the conjectured bound; everything the Fourier
measurement contributes is excess.

It remains to show that the Fourier outcomes are correlated. In local
dimension \(d\) their conditional laws are
\begin{equation}
 r_0^{(d)}(k)=\frac1d,\qquad
 r_1^{(d)}(k)=\frac{1+\cos(2\pi k/d)}d,
 \label{eq:app-fourier-channel}
\end{equation}
the first being the Fourier law of \(|0\rangle\) and the second following by
squaring the Fourier amplitude of \((|1\rangle+|2\rangle)/\sqrt2\). At the
output \(k=0\) the two conditional probabilities are \(1/d\) and \(2/d\), so
the Fourier joint law \(P^R\) satisfies
\begin{align}
 P^R_{00}
 &=\frac12\cdot\frac1m\cdot\frac1n+\frac12\cdot\frac2m\cdot\frac2n
 =\frac5{2mn}, \label{eq:app-jointcell}\\
 P^R_{0\bullet}&=\frac3{2m},\qquad
 P^R_{\bullet0}=\frac3{2n}, \label{eq:app-marginalcell}
\end{align}
and consequently
\begin{equation}
 P^R_{00}-P^R_{0\bullet}P^R_{\bullet0}=\frac1{4mn}>0 .
 \label{eq:app-separation}
\end{equation}
The Fourier joint law is therefore not a product law.

For an explicit margin put \(M=P^R-P^R_A\otimes P^R_B\). Its entries sum to
zero, so its total positive mass equals
\(\TV(P^R,P^R_A\otimes P^R_B)\); its \((0,0)\) entry is positive
by~\eqref{eq:app-separation}, whence
\(\TV(P^R,P^R_A\otimes P^R_B)\ge1/(4mn)\). Pinsker's inequality in natural
logarithms gives
\begin{equation}
 \begin{aligned}
 \cmir&=D(P^R\Vert P^R_A\otimes P^R_B)\\
 &\ge2\,\TV(P^R,P^R_A\otimes P^R_B)^2\ \ge\ \frac1{8m^2n^2}.
 \end{aligned}
 \label{eq:app-pinsker}
\end{equation}
Combining \eqref{eq:app-qmi}, \eqref{eq:app-computational}
and~\eqref{eq:app-pinsker} gives \eqref{eq:universalexcess}.
\end{proof}

\begin{proof}[Proof of Corollaries~\ref{cor:fullrank} and~\ref{cor:strengthenedfullrank}]
Mix the respective rank-two witness \(\rho\) with white noise:
\(\rho_\varepsilon=(1-\varepsilon)\rho+\varepsilon(I_A/d_A)\otimes(I_B/d_B)\).
For \(0<\varepsilon<1\) this state is full rank and separable. At fixed
measurements the gap is continuous in \(\varepsilon\) and is strictly negative
at zero, so it stays negative for all sufficiently small positive \(\varepsilon\).
\end{proof}

\subsection{The all-MUB qubit--qudit tail}

\begin{proof}[Proof of Theorem~\ref{thm:qubittail}]
Only the mutual unbiasedness \(|\langle q_j,r_k\rangle|^2=n^{-1}\) is used.
For an equal mixture of two pure rays of absolute overlap \(t\) the two
nonzero eigenvalues are \((1\pm t)/2\), so the spectrum contributes
\(h(t)\). The Alice, Bob and joint overlaps in \eqref{eq:omega} are
\(2^{-1/2}\), \(n^{-1/2}\) and \((2n)^{-1/2}\) respectively, whence
\begin{equation}
 M_n:=\qmi_{\omega_n}
 =h(1/\sqrt2)+h(1/\sqrt n)-h(1/\sqrt{2n}).
 \label{eq:app-tail-qmi}
\end{equation}

The \((Z_A,Q_B)\) and \((X_A,R_B)\) tables are identical up to outcome
labels: in the first, the \(|0,q_0\rangle\) ray is a point mass at
\((0,0)\) while the \(|+,r_0\rangle\) ray is uniform on all \(2n\) cells; in
the second the two roles are exchanged. Linearity of the mixture gives, in
either table,
\begin{equation}
 \begin{gathered}
 p_{00}=\frac{2n+1}{4n},\qquad p_{10}=\frac1{4n},\\
 p_{aj}=\frac1{4n}\quad(a\in\{0,1\},\ j>0),
 \end{gathered}
 \label{eq:app-tail-table}
\end{equation}
with row marginal \((3/4,1/4)\), column-zero marginal \((n+1)/(2n)\) and all
other column marginals \(1/(2n)\). Each table therefore has mutual
information
\begin{equation}
 T_n=
 \frac{2n+1}{4n}\log\frac{2(2n+1)}{3(n+1)}
 +\frac1{4n}\log\frac2{n+1}
 +\frac{n-1}{4n}\log\frac43 .
 \label{eq:app-tail-table-mi}
\end{equation}
Put \(G_n=M_n-2T_n\); CQC is violated exactly when \(G_n<0\).

Both \eqref{eq:app-tail-qmi} and \eqref{eq:app-tail-table-mi} are explicit
elementary expressions in \(n\). Evaluating them at \(n=164\) gives the
strict inequality
\begin{equation}
 G_{164}<-5.6488839398977894\times10^{-5}<0 .
 \label{eq:app-tail-basepoint}
\end{equation}
The same evaluation at \(n=163\) returns
\(G_{163}>1.6668129404683359\times10^{-5}>0\); this adjacent value is
recorded only to show that the threshold is not loose, and no
minimal-dimension claim is made.

The rest is analytic. Set \(x=1/n\) and let \(M(x)\), \(T(x)\) and
\(G(x)=M(x)-2T(x)\) denote the right-hand sides
of~\eqref{eq:app-tail-qmi} and~\eqref{eq:app-tail-table-mi}. Direct
differentiation gives
\begin{equation}
 T'(x)=\frac14\log\frac{x(2+x)}{(1+x)^2},
\end{equation}
while \(h'(t)=-\operatorname{arctanh}t\) gives
\begin{equation}
 \begin{aligned}
 M'(x)&=\frac{2^{-1/2}\operatorname{arctanh}\sqrt{x/2}
 -\operatorname{arctanh}\sqrt x}{2\sqrt x}\\
 &>-\frac1{2(1-x)},
 \end{aligned}
\end{equation}
the inequality following from \(\operatorname{arctanh}t/t<1/(1-t^2)\) for
\(0<t<1\) after discarding the positive first term. For \(0<x\le1/164\) the
right-hand side is at least \(-82/163\), while
\begin{equation}
 \frac{(1+x)^2}{x(2+x)}=\frac1{1-(1+x)^{-2}}\ \ge\ \frac{27225}{329}>4,
\end{equation}
the displayed ratio being decreasing in \(x\), which is what makes the
endpoint the worst case. Consequently
\begin{equation}
 G'(x)>-\frac{82}{163}+\log2>0,
\end{equation}
where \(\log2>7/12>82/163\) follows from the two right-endpoint rectangles
for \(\int_0^1(1+t)^{-1}\,dt\). Thus \(G\) is \(C^1\) and increasing on
\((0,1/164]\), so \(1/n\le1/164\) gives
\(G_n\le G_{164}<0\) for every \(n\ge164\). Exchanging the parties proves
the transposed statement.
\end{proof}

\subsection{The dimension classification}

\begin{proof}[Proof of Theorem~\ref{thm:strengthened}]
We first record the mechanism common to all four constructions. Suppose unit
vectors \(u,v\in\mathbb C^n\) have disjoint supports in \emph{both} of Bob's
measurement bases \(Q_B,R_B\), and put
\begin{equation}
 \begin{gathered}
 \rho(u,v)=\tfrac12|0,u\rangle\langle0,u|+\tfrac12|+,v\rangle\langle+,v|,\\
 |+\rangle=\frac{|0\rangle+|1\rangle}{\sqrt2}.
 \end{gathered}
 \label{eq:app-mechanism-state}
\end{equation}
Then either Bob measurement identifies the latent branch perfectly. Since
\(\langle u,v\rangle=0\), the two-ray spectra and the two binary channels on
Alice give
\begin{equation}
 \begin{gathered}
 \qmi=h(1/\sqrt2),\\
 I(Z_A:Q_B)=I(X_A:R_B)=h(1/2)-\tfrac12\log2,
 \end{gathered}
\end{equation}
so that the gap is the dimension-free constant
\begin{equation}
 G_\circ=\frac32\log\frac32
 -\frac1{\sqrt2}\operatorname{arctanh}\frac1{\sqrt2}
 \ <\ -\frac{17}{2688}\ <\ 0 .
 \label{eq:app-mechanism-gap}
\end{equation}
The numerical bound is elementary: the series for
\(\operatorname{arctanh}\) gives
\begin{equation}
 \frac1{\sqrt2}\operatorname{arctanh}\frac1{\sqrt2}
 >\frac12+\frac1{12}+\frac1{40}+\frac1{112}=\frac{1037}{1680},
\end{equation}
and the alternating series for the logarithm gives
\begin{equation}
 \log\frac32<\frac12-\frac18+\frac1{24}-\frac1{64}+\frac1{160}
 =\frac{391}{960}.
\end{equation}
So it suffices to exhibit, in each dimension, one MUB pair and one pair
\(u,v\) with disjoint supports in both bases.

\emph{Composite \(n\).} Write \(n=ab\) with \(a,b\ge2\) and set
\begin{equation}
 u=\frac1{\sqrt a}\sum_{t=0}^{a-1}|bt\rangle,\qquad v=XZu,
\end{equation}
where \(X|j\rangle=|j+1\rangle\) and \(Z|j\rangle=e^{2\pi ij/n}|j\rangle\).
Their computational supports are the disjoint cosets \(b\mathbb Z_n\) and
\(1+b\mathbb Z_n\). In the canonical Fourier basis,
\begin{equation}
 \langle r_k|u\rangle=b^{-1/2}\mathbf 1_{a\mid k},\qquad
 \langle r_k|v\rangle=e^{-2\pi ik/n}b^{-1/2}\mathbf 1_{a\mid(1-k)},
\end{equation}
so their supports there are the disjoint cosets \(a\mathbb Z_n\) and
\(1+a\mathbb Z_n\). Hence \eqref{eq:app-mechanism-gap} applies in every
composite dimension \(n\ge4\).

\emph{The Fourier two-spoke tail.} For arbitrary \(n\) let \(Q_B\) be
computational, \(R_B\) canonical Fourier, \(\theta=\pi/n\), and
\begin{equation}
 u_n=\frac{q_0-e^{i\theta}q_1}{\sqrt2},\qquad
 v_n=\frac{r_0-e^{-i\theta}r_1}{\sqrt2}.
\end{equation}
Here the supports are not disjoint, so \eqref{eq:app-mechanism-gap} does not
apply and the gap must be computed. For \(\rho(u_n,v_n)\) put \(x=1/n\) and
\begin{equation}
 c(x)=2\sqrt x\sin\frac{\pi x}2,\qquad s(x)=x(1-\cos\pi x).
\end{equation}
The Alice, Bob and global ray overlaps give
\(M(x)=\qmi=h(1/\sqrt2)+h(c)-h(c/\sqrt2)\). The two tables agree up to
relabelling, and with
\(A=(1+s)/4\), \(D=s/4\), \(C=(1+2s)/4\) their common mutual information is
\begin{equation}
 \begin{aligned}
 T(s)&=H(3/4,1/4)\\
 &\quad+2\left(A\log A+D\log D-C\log C\right)\\
 &\quad-\frac{1-2s}2\log2 .
 \end{aligned}
\end{equation}
Writing \(G(x)=M(x)-2T(s(x))\), evaluation at the single point \(x=1/12\)
gives
\begin{equation}
 G(1/12)<-0.0008829172049654<0 .
 \label{eq:app-spoke-basepoint}
\end{equation}
The endpoint controls the whole tail: differentiation gives
\begin{equation}
 \begin{aligned}
 G'(x)&=-\left(\operatorname{arctanh}c
 -\frac1{\sqrt2}\operatorname{arctanh}\frac c{\sqrt2}\right)c'\\
 &\quad+s'\log\frac{(1+2s)^2}{4s(1+s)},
 \end{aligned}
\end{equation}
and on \(0<x\le1/12\) elementary Taylor bounds give
\((1296/859)\pi^2x^2\) as an upper bound for the magnitude of the negative
term while the positive term exceeds \((148/25)\pi^2x^2\). Both follow from
the leading behaviour \(c(x)=\pi x^{3/2}+O(x^{7/2})\),
\(s(x)=\tfrac12\pi^2x^3+O(x^5)\) together with
\(\operatorname{arctanh}c-\tfrac1{\sqrt2}\operatorname{arctanh}\tfrac c{\sqrt2}
=\tfrac c2+O(c^3)\), the remainders being controlled on the whole interval by
their values at the endpoint \(x=1/12\). Hence \(G'(x)>0\) there, so
\(G(1/n)\le G(1/12)<0\) for every \(n\ge12\).

\emph{The remaining primes seven and eleven.} At \(n=7\) let
\(\zeta=e^{i\pi/3}\) and let \(P_7\) have entries \(\zeta^{e_{jk}}\) with
exponent rows
\begin{equation}
 \begin{gathered}
 0000000,\quad 0145331,\quad 0413531,\quad 0531413,\\
 0354113,\quad 0331145,\quad 0113354 .
 \end{gathered}
\end{equation}
Reduction modulo \(\zeta^2-\zeta+1\) gives \(P_7P_7^*=7I\), so
\(\overline{P_7}/\sqrt7\) has orthonormal columns forming a MUB partner of
the computational basis~\cite{Petrescu1997}. Set
\begin{equation}
 \begin{gathered}
 u=\frac{\zeta q_0-q_1}{\sqrt2},\\
 v=\sqrt{\frac27}\left(q_2+\frac{1-2\zeta}2q_3
 +\frac{2\zeta-1}2q_5-q_6\right).
 \end{gathered}
\end{equation}
Their computational supports are \(\{0,1\}\) and \(\{2,3,5,6\}\), and
multiplication by \(P_7\) gives supports \(\{0,2,3,4,5\}\) and \(\{1,6\}\)
in the second basis; both pairs are disjoint.

At \(n=11\) take the symmetric Nicoar\u a matrix \(N_{11A}\) of the
catalogue~\cite{TadejZyczkowski2006} and
\begin{equation}
 a=-\frac34-i\frac{\sqrt7}4,\qquad a^2+\frac32a+1=0 .
\end{equation}
With MUB columns \(\overline{N_{11A}}/\sqrt{11}\), use
\begin{equation}
 \begin{gathered}
 u=\frac{q_2-q_3}{\sqrt2},\\
 v=\frac{\begin{gathered}
 q_0+(-4a-2)q_1\\
 {}+(-4a-3)(q_5+q_6+q_9+q_{10})\\
 {}+(1-a)(q_7+q_8)
 \end{gathered}}{\sqrt{44}} .
 \end{gathered}
\end{equation}
Arithmetic in \(\mathbb Q[a]/(a^2+3a/2+1)\) gives
\(N_{11A}N_{11A}^*=11I\), computational supports \(\{2,3\}\) and
\(\{0,1,5,6,7,8,9,10\}\), and second-basis supports \(\{2,3,4,6,7,10\}\) and
\(\{0,1\}\); again both pairs are disjoint, so
\eqref{eq:app-mechanism-gap} applies at \(n=7\) and \(n=11\).

Every integer \(n\ge3\) other than \(3\) and \(5\) is now covered:
composites by the subgroup construction, the primes \(7\) and \(11\) by the
exact Hadamard witnesses, and every prime at least thirteen by the two-spoke
tail. Exchanging the parties proves the transposed statement.
\end{proof}

\section{Additional two-ray analysis}\label{app:two-ray-details}

\subsection{Where the qutrit constant comes from}\label{sec:stratum}

\begin{proposition}[Doubly degenerate stratum]\label{prop:stratum}
If \(W\) contains one \(Q\) vector and one \(R\) vector, then every other
projected \(Q\) vector lies along the single ray of \(W\) orthogonal to that
\(Q\) vector, so \(A_Q\) collapses to a rank-one projector, and likewise
\(A_R\). Their Bloch axes then have overlap
\(2|\langle q_j,r_k\rangle|^2-1=2/d-1\) \emph{by unbiasedness alone}, whence
\begin{equation}
 \lmax(A_W)=2-\frac2d,\qquad S=4-\frac4d,\qquad e_2=\frac{4(d-1)}{d^2}
\end{equation}
exactly, in every dimension.
\end{proposition}

\begin{proof}
Say \(q_j,r_k\in W\). For \(j'\ne j\), \(\langle q_j|P_Wq_{j'}\rangle
=\langle q_j,q_{j'}\rangle=0\) because \(P_Wq_j=q_j\); so every other
projected \(Q\) vector lies in the one-dimensional space \(W\cap q_j^\perp\),
spanned by some unit \(x\). The Bloch vector of \(q_j\) itself is antipodal
to that of \(x\), and antipodal Bloch vectors have equal outer products, so
\(A_Q=\bigl(\sum_j\nu_j\bigr)\hat n_x\hat n_x^{\mathsf T}
=\hat n_x\hat n_x^{\mathsf T}\), a rank-one projector; likewise
\(A_R=\hat n_y\hat n_y^{\mathsf T}\) with \(y\) spanning \(W\cap r_k^\perp\).
In a two-dimensional space, orthocomplementation inside \(W\) preserves
absolute overlaps, so \(|\langle x,y\rangle|=|\langle q_j,r_k\rangle|\) and
\(\hat n_x\!\cdot\!\hat n_y=2|\langle q_j,r_k\rangle|^2-1=2/d-1\) by
unbiasedness. A sum of two rank-one projectors with axis overlap \(k\) has
eigenvalues \(1+|k|,\,1-|k|,\,0\), giving
\(\lmax=2-2/d\), \(S=2\lmax=4-4/d\) by \eqref{eq:lmax}, and
\(e_2=(1+|k|)(1-|k|)=4(d-1)/d^2\).
\end{proof}

At \(d=3\) this is \(S=8/3\): the sharp qutrit constant is forced by
unbiasedness with no computation at all, which is a cleaner derivation than
any determinant evaluation. At \(d=5\) it is only \(16/5=3.2\), below the
attained \((14+2\sqrt5)/5\), so that stratum is \emph{not} extremal there ---
an early sign that \(d=5\) is genuinely harder than \(d=3\).

\subsection{The two-ray family at \texorpdfstring{$2\times5$}{2x5}}\label{sec:n5}

\begin{theorem}[Exact saturator at \(d=5\)]\label{thm:n5}
Since \(F_5^2\) is the parity map \(j\mapsto-j\), the \(\pm1\) eigenvectors of
\(F_5\) lie in the three-dimensional parity-symmetric subspace spanned by
\(e_0,(e_1+e_4)/\sqrt2,(e_2+e_3)/\sqrt2\), on which \(F_5\) acts as the exact
real symmetric involution
\begin{equation}
 M=\frac1{\sqrt5}
 \begin{pmatrix}
 1&\sqrt2&\sqrt2\\
 \sqrt2&\frac{\sqrt5-1}2&-\frac{\sqrt5+1}2\\
 \sqrt2&-\frac{\sqrt5+1}2&\frac{\sqrt5-1}2
 \end{pmatrix},\qquad M^2=I,
\end{equation}
with \(-1\) simple and \(+1\) of multiplicity two. Fix the eigenvectors
explicitly, each of norm \(\sqrt2\): with
\(A=\sqrt{5-\sqrt5}\) and \(B=\sqrt{5+\sqrt5}\),
\begin{equation}
\begin{gathered}
 a=\frac1A\bigl(1-\sqrt5,\,1,\,1,\,1,\,1\bigr),\\
 \xi_1=\frac1{\sqrt2}\bigl(0,\,1,\,-1,\,-1,\,1\bigr),\\
 \xi_2=\frac1B\bigl(1+\sqrt5,\,1,\,1,\,1,\,1\bigr),
\end{gathered}
\end{equation}
so that \(Fa=-a\), \(F\xi_1=\xi_1\), \(F\xi_2=\xi_2\) and
\(\xi_1\perp\xi_2\). Put
\(b=\cos t\,\xi_1+\sin t\,\xi_2\); then \(u=(a+b)/2\), \(v=(b-a)/2\) are orthonormal
for every \(t\) and satisfy \(Fu=v\), \(Fv=u\). At the exact critical angle
\begin{equation}
\begin{gathered}
 \tan^2t=\frac{5-\sqrt5}{10},\quad\tan t<0,
 \\
 \cos^2t=\frac{15+\sqrt5}{22},\ \ \sin^2t=\frac{7-\sqrt5}{22},
\end{gathered}
\end{equation}
one has, exactly,
\begin{equation}
\begin{gathered}
 \Delta(p^Q,q^Q)=\Delta(p^R,q^R)=\frac{7+\sqrt5}5,
 \\
 S=\frac{14+2\sqrt5}5\approx3.694,
\end{gathered}
\end{equation}
together with
\(\TV_Q=\TV_R=\sqrt{(7+\sqrt5)/10}\),
\(\mathrm{BC}_Q=\mathrm{BC}_R=(5-\sqrt5)/10\) and
\(e_2(A_W)=(16-4\sqrt5)/25\).
\end{theorem}

The Fourier laws are verified by the direct Born rule, not by assuming the
exchange. The angle is gauge dependent: it refers to the specific basis
\(\{\xi_1,\xi_2\}\) fixed above, and a different orthonormal basis of the \(+1\)
eigenspace relabels \(t\) while leaving every quantity below unchanged.
Theorem~\ref{thm:n5} is a \emph{sharpness} statement: it proves the value is
attained, not that it is maximal. Attainment alone, however, settles three
competing routes, each comparison decided by rigorous interval arithmetic
against \(S_\star\).

\begin{corollary}[Three route refutations at \(d=5\)]\label{cor:refute}
The exact example excludes the following three arguments for closing the family.
\begin{enumerate}
\item The quadratic-majorant route in Remark~\ref{rem:whyenvelope}
 cannot close the family: the attained value exceeds \(8/3\).
\item The elementary pointwise bound \((p-q)^2/(p+q)\le|p-q|\) gives
 \(S\le2(\TV_Q+\TV_R)\); at the saturator that bound equals
 \(4\sqrt{(7+\sqrt5)/10}=3.844179>S_\star\), even though the true \(S\) there
 is only \(3.694427\).
\item \(\Delta\le2-\mathrm{BC}^2\) gives
 \(S\le4-\mathrm{BC}_Q^2-\mathrm{BC}_R^2=(17+\sqrt5)/5=3.847214>S_\star\).
\end{enumerate}
\end{corollary}

What survives is the eigenvalue-free criterion \eqref{eq:e2floor}. Write
\(\alpha_1\ge\alpha_2\ge\alpha_3\) for the spectrum of \(A_Q\) with unit
eigenvectors \(\hat e_1,\hat e_2,\hat e_3\), and
\(\beta_1\ge\beta_2\ge\beta_3\) for that of \(A_R\) with unit eigenvectors
\(\hat f_1,\hat f_2,\hat f_3\); let
\(\kappa=|\langle\hat e_1,\hat f_1\rangle|\) be the overlap of the two
dominant axes. One has the elementary bound
\begin{equation}
 X\ \ge\ 1-\alpha_1\bigl(\beta_1\kappa^2+\beta_2(1-\kappa^2)\bigr)
 -(1-\alpha_1)\beta_1,
 \label{eq:refined}
\end{equation}
and symmetrically, using \(\sum_i\langle\hat e_1,\hat f_i\rangle^2=1\) and
\(1-\alpha_1\le D_Q/2\). Combined with \(e_2\ge(D_Q+D_R)/2\), minimising
\(\max\{(D_Q+D_R)/2,\ \eqref{eq:refined}\}\) --- over genuine two-dimensional
subspaces \(W\), parameterised by a complex \(d\times2\) frame and not over the
five spectral quantities treated as free variables --- gives \(0.5000\) at
\(d=3\) and,
under an adversarial search (400 multistart Nelder--Mead and Powell, then
differential evolution), \(0.2253481\) at \(d=5\), still clearing the floor
\(0.169681\) by \(0.0557\). This is a candidate route for the two-ray family in both dimensions and reduces \(2\times5\) from a twelve-dimensional Grassmannian to
five spectral quantities. Proving it over the Grassmannian is open, and
\(2\times5\) remains the single unresolved case for this family.

\begin{remark}
The cruder version of \eqref{eq:refined}, obtained by charging every
non-leading cross term the full value one, reads
\(X\ge1-\kappa^2-(D_Q+D_R)/2\). It suffices at \(d=3\), where a targeted
search gives \(\kappa\approx0.358\) in the residual region, but it fails at
\(d=5\), where \(\kappa\) reaches \(0.877\) at a point whose true \(X\) is
\(0.287\). Recording this matters: plain random sampling almost never enters
that thin region and would have suggested \(\kappa\approx0.088\).
\end{remark}

\section{Further qutrit sampling diagnostics}\label{app:sampling-details}

\subsection{The sample minimum improves logarithmically}\label{sec:scaling}

Concentration in high dimension makes the extreme order statistic move very
slowly. Running \(10^6\) Hilbert--Schmidt samples and recording the running
minimum, and repeating the whole experiment from eight independent seeds,
gives the median trace

\begin{center}
\begin{tabular}{rr}
\hline\hline
\(N\) & median running \(\min G\)\\
\hline
\(10^2\) & 0.23334\\
\(10^3\) & 0.20096\\
\(10^4\) & 0.18539\\
\(10^5\) & 0.16337\\
\(10^6\) & 0.14249\\
\hline\hline
\end{tabular}
\end{center}

Over four decades the median minimum falls by \(0.0909\), and a least-squares
fit in \(\log N\) gives \(0.0219\) nats per decade. Extrapolating that rate
linearly and reaching \(G=-0.0139\) from \(0.14249\) needs a further
\(0.15642/0.02193\approx7.1\) decades, that is \(N\sim10^{13.1}\).

No single such number should be quoted as if it were determined. The slope is
fitted to an extreme order statistic over only four decades, and it is very
seed sensitive: across the eight seeds the fitted slopes range from
\(0.0149\) to \(0.0309\) nats per decade, median \(0.0223\), and
correspondingly
\begin{equation}
 N\ \sim\ 10^{10.4}\ \text{to}\ 10^{17.3}, \qquad\text{median }10^{12.9}.
 \label{eq:naiveextrap}
\end{equation}
This is the band shown in Figure~\ref{fig:concentration}(b).

A better-grounded estimate uses the shape of the tail instead of a slope.
Writing \(z(N)=(\text{mean}-\min_N)/\text{std}\) for the standardised depth
of the running minimum, a Gaussian left tail would give \(z=\sqrt{2\log N}\).
Fitting the median trace instead gives
\begin{equation}
 z(N)\ \approx\ 0.805\,\sqrt{2\log N}-0.129 ,
 \label{eq:evfit}
\end{equation}
so the tail is \emph{lighter} than Gaussian and the minimum descends more
slowly than a Gaussian tail would allow --- which makes matters worse, not
better. The witness sits at \(z=7.30\), and \eqref{eq:evfit} reaches that at
\(N\sim10^{18.5}\); the pure-Gaussian reference \(N=e^{z^2/2}\) would already
demand \(10^{11.6}\).

The three routes disagree by many orders of magnitude, which is the honest
state of the estimate, but they agree on everything that matters: the
smallest of them exceeds the \(10^7\) samples actually performed by more than
three orders of magnitude, and the best-motivated one by more than eleven. We
therefore quote the conclusion as a range, not a number. The same experiment
at \(5\otimes5\) gives running minima \(0.37177\) at \(N=10^3\) and
\(0.36765\) at \(N=10^4\), again with no negative sample.

\subsection{The violating set hugs the low-rank boundary}\label{sec:volume}

The qutrit witness has rank three out of nine, and
survives white noise only up to
\begin{equation}
 \varepsilon_c=0.020982692465\ldots,
\end{equation}
so its full-rank neighbourhood is thin. Taking the interior violating point
\(\rho=(1-\varepsilon)\rho_\star+\varepsilon I/9\) with
\(\varepsilon=\varepsilon_c/2\) (full rank, \(G=-0.00538\)) and moving in
\(500\) random traceless-Hermitian directions until either \(G\) turns
nonnegative or the state-space boundary is reached, the distance travelled has
median \(0.002950\), mean \(0.002978\), maximum \(0.005060\) in
Hilbert--Schmidt norm. In \(469\) of the \(500\) directions the
\emph{state-space boundary} was hit first: the violating region is pressed
against the low-rank face, exactly where random ensembles place no mass.
Figure~\ref{fig:geometry} shows both halves of this picture directly.

\begin{figure*}[t]
\centering
\includegraphics[width=0.92\textwidth]{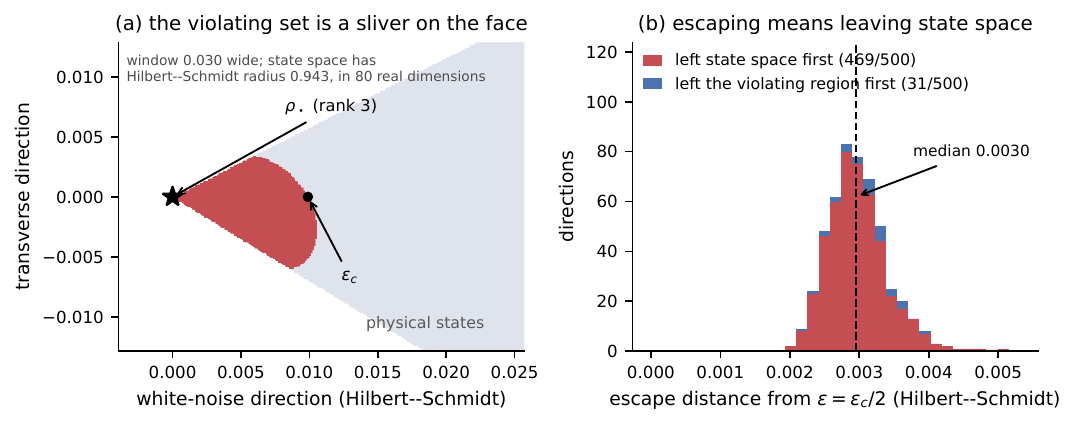}
\caption{Where the counterexamples live. \textbf{(a)} A two-dimensional
affine slice of state space through the witness \(\rho_\star\). The
horizontal axis is the white-noise direction \(I/9-\rho_\star\); the vertical
axis is a transverse direction, taken to be the one realising the
\emph{median} escape distance of panel~(b), so the slice is representative
rather than chosen for effect. Light grey is the physical region, red the
violating region \(G<0\). The violating set is a lens pressed against the
low-rank face, with \(\rho_\star\) itself at its vertex on the boundary; along
the horizontal axis it terminates at \(\varepsilon_c\). The scale annotation
is essential to reading the panel: the window is \(0.030\) across, against a
state-space Hilbert--Schmidt radius of \(0.943\), and two of eighty real
dimensions are shown, so the apparent area of the red region says nothing
about its volume. \textbf{(b)} Escape distances from the interior violating
point at \(\varepsilon=\varepsilon_c/2\) in \(500\) random traceless-Hermitian
directions, split by whether the walk left state space or left the violating
region first. The state-space boundary is reached first in the large majority
of directions --- this run gives \(469/500\), inside the \(93\)--\(95\%\)
band quoted above.}
\label{fig:geometry}
\end{figure*}

These four figures are Monte-Carlo statistics and therefore seed dependent. An
independent reimplementation reproduces the median and mean within one to two
per cent and the minimum closely; the maximum, being an extreme order
statistic, fluctuates by more, and the boundary-first fraction ranges over
roughly \(93\)--\(95\%\) across seeds, the value \(469/500\) quoted above
sitting inside it. The qualitative conclusion --- that the
violating region is pressed against the low-rank face in the large majority of
directions --- is robust to reseeding. The geometric quantities
(\(\varepsilon_c\), the parameter count, the Hilbert--Schmidt radius) are
seed independent.

For scale, the state space of \(\mathbb C^3\otimes\mathbb C^3\) has \(80\) real
parameters and Hilbert--Schmidt radius \(\sqrt{1-1/9}=0.9428\) about the
maximally mixed state. A crude ball heuristic gives relative volume
\begin{equation}
 \left(\frac{0.002950}{0.9428}\right)^{80}\approx4.3\times10^{-201},
\end{equation}
i.e.\ order \(10^{201}\) uniform samples to hit the region once. We stress
that this is a heuristic scale estimate, not a bound: the inradius gives a
lower bound on volume, and the violating set is certainly not a ball. It is
offered only to indicate the order of magnitude of the mismatch between
\(10^7\) and what would be required.

\subsection{The measurements must be aligned too}\label{sec:align}

Even with the right state, the measurement frame must be nearly aligned. Under
a local unitary rotation
\(\rho_\star\mapsto(e^{i\theta H_A}\otimes e^{i\theta H_B})\rho_\star(\cdot)^\dagger\)
with random traceless Hermitian generators, the violation survives only up to
a critical angle whose distribution over \(600\) random directions has median
\(0.156\) rad \(=8.93^\circ\) and minimum \(0.098\) rad. Both reproduce across
seeds to better than three per cent; the maximum is an extreme order statistic
and fluctuates, independent runs placing it between \(0.39\) and \(0.54\) rad.
The estimate below uses the median, which is the stable quantity;
Figure~\ref{fig:alignment} shows the whole survival curve. Since
\(\dim SU(3)\times SU(3)=16\), the fraction of local frames within the median
tolerance scales as
\begin{equation}
 \left(\frac{0.156}{\pi}\right)^{16}\approx1.4\times10^{-21}.
\end{equation}
This is a heuristic scale estimate, not a probability bound. It models
the tolerance region as isotropic and does not establish independence
from the state-space estimate of Sec.~\ref{sec:volume}.

\begin{figure}[t]
\centering
\includegraphics[width=\columnwidth]{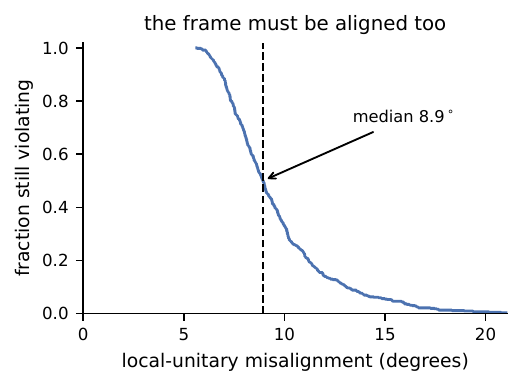}
\caption{The measurement frame must be aligned too. Fraction of \(600\)
random local-unitary rotation directions for which \(\rho_\star\) still
violates CQC, against the rotation angle. Half the directions have lost the
violation by about nine degrees, and essentially none survives past twenty.
The curve is a distribution over directions at each angle, not a bound.}
\label{fig:alignment}
\end{figure}

\bibliography{refs}

\end{document}